\documentclass[sigconf,review=false]{acmart}

\usepackage{appendix}

\usepackage{booktabs}
\usepackage[utf8]{inputenc} 
\usepackage[T1]{fontenc}    

\usepackage{url}            
\usepackage{booktabs}       
\usepackage{amsfonts}       
\usepackage{nicefrac}       
\usepackage{microtype}      
\usepackage{xcolor}         

\usepackage{algorithmic}
\usepackage{graphicx}
\usepackage{textcomp}
\usepackage{mathtools}

\usepackage{siunitx}

\usepackage{bm}
\usepackage{multirow}
\usepackage[caption=false]{subfig}

\usepackage{amsmath}
\usepackage{algorithmic}
\usepackage{tabularx}
\usepackage{microtype}
\usepackage{stmaryrd}
\usepackage{xspace}

\usepackage[ruled,vlined,linesnumbered]{algorithm2e}
\usepackage{wrapfig}
\usepackage[normalem]{ulem}
\usepackage[subtle,mathspacing=normal]{savetrees}
\usepackage[strings]{underscore}
\usepackage{etoolbox}

\newcommand{\zerodisplayskips}{%
  \setlength{\abovedisplayskip}{3pt}%
  \setlength{\belowdisplayskip}{3pt}%
  \setlength{\abovedisplayshortskip}{3pt}%
  \setlength{\belowdisplayshortskip}{3pt}}
\appto{\normalsize}{\zerodisplayskips}
\appto{\small}{\zerodisplayskips}
\appto{\footnotesize}{\zerodisplayskips}

\usepackage[english]{babel}
\newtheorem{example}{example}
\newtheorem{remark}{remark}
\newtheorem{problem}{problem}

\newcommand{\srs}{R} 
\newcommand{\Trec}{\alpha}
\newcommand{\Tdur}{\beta}

\newcommand{\system}{(F,G,x_0)} 
\newcommand{\states}{\Vec{x}} 
\newcommand{\controls}{\Vec{u}} 
\newcommand{\recencoding}{c^{\varphi,rec}_t} 
\newcommand{\durencoding}{c^{\varphi,dur}_t} 
\newcommand{\recobj}{c^{\varphi,rec}_0} 
\newcommand{\durobj}{c^{\varphi,dur}_0} 
\newcommand{\probname}{resilient STL control\xspace} 
\newcommand{\probnamecap}{Resilient STL Control\xspace} 
\newcommand{\rescontroller}{resilient STL controller\xspace}
\newcommand{\name}{ResilienC\xspace}

\begin{document}
\title[An STL-based Approach to Resilient Control for Cyber-Physical Systems]{
An STL-based Approach to Resilient Control\break for Cyber-Physical Systems
}

%

\author{Hongkai Chen}
\affiliation{
  \institution{Department of Electrical and Computer Engineering,\\Stony Brook University}
  \city{Stony Brook}
  \country{USA}}
\email{hongkai.chen@stonybrook.edu}
\author{Scott A. Smolka}
\affiliation{
  \institution{Department of Computer Science,\\Stony Brook University}
  \city{Stony Brook}
  \country{USA}}
\email{sas@cs.stonybrook.edu}
\author{Nicola Paoletti}
\affiliation{
  \institution{Department of Informatics, \\King’s College London}
  \city{London}
  \country{UK}}
\email{nicola.paoletti@kcl.ac.uk}
\author{Shan Lin}
\affiliation{
  \institution{Department of Electrical and Computer Engineering,\\Stony Brook University}
  \city{Stony Brook}
  \country{USA}}
\email{shan.x.lin@stonybrook.edu}

\begin{abstract}


We present \emph{\name}, a framework for resilient control of Cyber-Physical Systems
subject to STL-based requirements.
\name utilizes a recently developed
formalism for specifying CPS resiliency in terms of sets of $(\mathit{rec},\mathit{dur})$ real-valued pairs, where $\mathit{rec}$ represents the system's capability to rapidly recover from a property violation (\emph{recoverability}), and $\mathit{dur}$ is reflective of its ability to avoid violations post-recovery (\emph{durability}).
We define the
\emph{resilient STL control problem} as one of \emph{multi-objective optimization}, where the recoverability and durability of the desired STL specification are maximized. When neither objective is prioritized over the other, the solution to the problem is a set of \emph{Pareto-optimal} system trajectories. We present a precise solution method to the \probname problem using a mixed-integer linear programming encoding and an \emph{a posteriori} $\epsilon$-constraint approach for efficiently retrieving the complete set of optimally resilient solutions. In \name, at each time-step, the optimal control action selected from the set of Pareto-optimal solutions by a \emph{Decision Maker} strategy realizes a form of \emph{Model  
Predictive Control}.
We demonstrate the practical utility of the \name framework on two significant case studies: autonomous vehicle lane keeping and deadline-driven, multi-region package delivery.

\end{abstract}

\maketitle

\section{Introduction}
Resiliency is of fundamental importance in Cyber-Physical Systems (CPS), as such systems are expected to fulfill safety- and mission-critical requirements even in the presence of external disturbances or internal faults.
Although various notions of resiliency have been proposed within a control setting~\cite{bouvier2021quantitative,zhu2011robust}, a general formal characterization has been lacking. Recently, Chen et al.~\cite{chen2022stl} used Signal Temporal Logic (STL)~\cite{maler2004monitoring} to formally reason about resiliency in CPS. 
Given an STL property $\varphi$ expressing a CPS requirement, the notion of resiliency introduced in~\cite{chen2022stl} permits violations of $\varphi$
as long as: 1)~the CPS quickly recovers from the violation, and then 2)~satisfies $\varphi$ for an extended period of time. These two requirements are called \emph{recoverability} and \emph{durability}, respectively. 

The results of~\cite{chen2022stl} naurally suggest the following problem of \emph{\probname}: find an optimal control strategy that
maximizes the system's resilience in terms of recoverability and durability.  These two objectives are often at odds with each other. For example, in the \emph{lane-keeping problem} (see Figure~\ref{fig:example1}), 
an aggressive control strategy can quickly return
the vehicle to the lane after a violation (good recoverability) but might fail to keep the vehicle in the lane for an extended period of time due to overshooting (poor durability). On the other hand, 
with a cautious strategy, 
the vehicle might take longer to re-enter the lane (poor recoverability) but subsequently manage to remain in the lane longer (good durability). In other words, without prioritizing one requirement
over the other, the aggressive and cautious strategies are \emph{mutually non-dominated} and, hence, equally resilient.

In this paper, we present a control framework called \emph{\name} (the `C' stands for control), where the \probname problem is formulated
as one of multi-objective optimization, designed to maximize both the recoverability and durability of the CPS. Unlike existing techniques for STL-based control~\cite{raman2014model,rodionova2021time} which focus on optimizing a single objective (e.g., spatial robustness in~\cite{raman2014model} and time robustness in~\cite{rodionova2021time}) and thus produce a single solution, our method results in a set of \emph{non-dominated, aka Pareto-optimal, solutions}. Such a method is also called \textit{a posteriori} 
 as it avoids making any \textit{a priori} assumptions about the relative importance of the two objectives (recoverability and durability). We achieve a Model Predictive Control (MPC) scheme with our method by deploying a \emph{Decision Maker} (DM) strategy that, at each time-step, selects the next optimal control action from among the set of Pareto-optimal solutions for execution by the plant. See Figure~\ref{fig:overall} for an overview of the \name framework. 

We solve the \probname problem in a precise manner, in that our method can retrieve the entire set of non-dominated resilient points. To do so, we focus on CPS with linear dynamics and encode the problem as one of mixed-integer linear programming (MILP). In particular, the solution set of the multi-objective problem is found by solving multiple single-objective MILP instances through an $\epsilon$-constraint approach~\cite{laumanns2006efficient}. 

From a theoretical standpoint, besides proving the correctness and characterizing the complexity of our algorithm, we establish an important relationship between our \probname problem and the time-robust STL control problem recently introduced in~\cite{rodionova2021time}.  Time-robust STL control seeks to optimize time robustness, i.e., the extent to which a trajectory can be shifted in time without affecting the satisfaction of the STL specification. We prove that any time-robust solution is also a resilient solution, making \probname a generalization of time-robust STL control.

\begin{figure}[t]
	\centering
	\includegraphics[width=\linewidth]{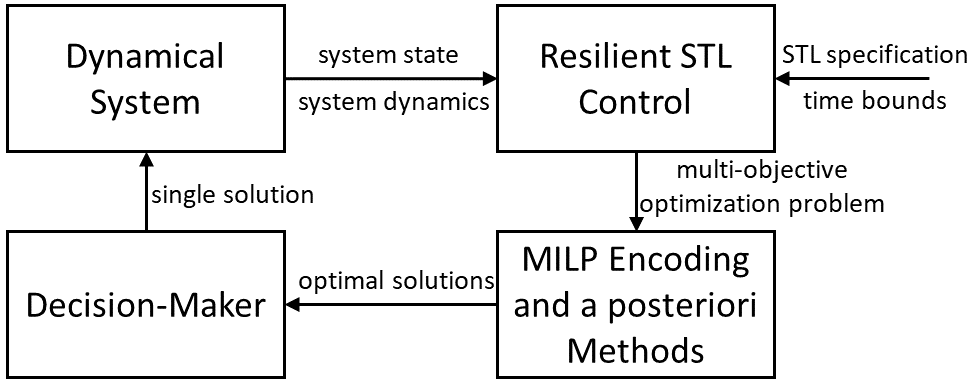}
	\caption{Overview of \name architecture.}
 \vspace{-4ex}
	\label{fig:overall}
\end{figure}

We evaluate \name on two case studies: lane keeping and deadline-driven multi-region package delivery. Our results clearly demonstrate the effectiveness of our solution method, which provides
a comprehensive view of the recoverability-durability tradeoff.  
Furthermore, we use the case studies to assess and compare the various DM strategies.

\vspace{-0.15ex}
In summary, our main contributions are the following.
\begin{itemize}
    \item We present \emph{\name}, a resilient control framework for CPS with STL-based requirements. We define the resilient control problem as one of multi-objective optimization such that the recoverability and durability metrics associated with the STL specifications are maximized and a set of  Pareto-optimal solutions is generated. We also propose various DM strategies for selecting a single optimal solution, used to generate MPC-based control actions. 
    To the best of our knowledge, we are the first to investigate a resilient control framework that co-optimizes recoverability and durability.
    \item We present a precise solution method, based on an MILP encoding and an a posteriori $\epsilon$-constraint approach, for efficiently retrieving the complete set of optimally resilient solutions.
    \item We prove that our resilient control framework is a generalization of time-robust STL control.
    \item We conducted two case studies for which we considered various control strategies that induce vastly different but equivalently resilient trajectories. We also illustrate the effects of multiple DM preferences on \name-based control. 
\end{itemize}


\section{Background}\label{sec:background}

In this section, we provide background on the syntax and semantics of both STL and the STL-based Resiliency Specifications of~\cite{chen2022stl}.
Let $\xi:\mathbb{T}\rightarrow\mathbb{R}^n$ be a signal where $\mathbb{T} = \mathbb{Z}_{\geq 0}$ is the (discrete) time domain. Given $t \in \mathbb{T}$ and interval $I$ on $\mathbb{T}$, $t + I$ is used to denote the set $\{t+t' \mid t' \in I\}$.

\subsection{Signal Temporal Logic}\label{subsec:stl}

STL is a logical formalism for specifying temporal properties over real-valued signals.
An STL atomic proposition $p\in \mathit{AP}$ is defined over $\xi$ and is of the form $p\equiv \mu(\xi(t)) \geq c$, 
$c\in\mathbb{R}$, and $\mu:\mathbb{R}^n\rightarrow\mathbb{R}$.  STL formulas $\varphi$ are defined according to the following grammar~\cite{donze2010robust}:
\begin{align*}
    \varphi \,::=\, 
    p\;|\; \neg\,\varphi\;|\; 
    \varphi_1\,\wedge\,\varphi_2 \;|\;     \varphi_1\,\mathbf{U}_I\,\varphi_2
\end{align*}
\noindent where $\mathbf{U}$ is the \emph{until} operator and\hspace{0.2em}$I$ is an interval on $\mathbb{T}$. Logical disjunction is derived from $\wedge$ and $\neg$ as usual, and
operators \emph{eventually} and \emph{always} are derived from $\mathbf{U}$ as usual: $\mathbf{F}_{I}\varphi = \top\,\mathbf{U}_I\,\varphi$ and $\mathbf{G}_{I}\varphi = \neg\, (\mathbf{F}_I \neg\,\varphi)$.
The satisfaction relation $(\xi,t)\models\varphi$, indicating whether $\xi$ satisfies $\varphi$ at time $t$, is defined as follows:
\vspace{0.25ex}
\begin{align*}
&(\xi, t) \models p &\Leftrightarrow &\hspace{2ex} \mu(\xi(t)) \geq c\vspace{0.02in}\\
&(\xi, t) \models \neg \varphi &\Leftrightarrow &\hspace{2ex} \neg ((\xi, t) \models \varphi)\vspace{0.02in}\\
&(\xi,t) \models \varphi_1 \wedge \varphi_2  &\Leftrightarrow&\hspace{2ex} (\xi,t) \models \varphi_1 \wedge (\xi,t) \models \varphi_2 \vspace{0.02in}\\
&(\xi,t) \models \varphi_1 \mathbf{U}_I \varphi_2 &\Leftrightarrow&\hspace{2ex} \exists\ t'\in t + I\ \text{s.t.}\ (\xi,t') \models \varphi_2 \wedge \\
&&&\hspace{6em}\forall\; t''\in [t,t'),\ (\xi,t'') \models \varphi_1 
\end{align*}
The \textit{characteristic function} $\chi$ of $\varphi$ relative to $\xi$ at time $t$ is defined such that $\chi(\varphi, \xi, t)=1$ when $(\xi,t)\models \varphi$ and $-1$ otherwise~\cite{donze2010robust}.

STL admits a quantitative semantics called (space) robustness~\cite{donze2010robust} that quantifies the extent to which $\xi$ satisfies $\varphi$ at time $t$. Its absolute value can be seen as the distance of $\xi$ from the set of trajectories satisfying (positive value) or violating (negative value) $\varphi$.

STL also admits
a
quantitative semantics called \emph{time robustness}~\cite{donze2010robust}, which is used to quantify the extent to which a trajectory can be shifted in time without affecting the satisfaction (or violation) of the STL specification.  Its definition is given in terms of the real-valued function $\theta^+$:
\begin{align*}
&\theta^+(p,\xi,t)&=&\hspace{1em}\chi(p, \xi,t)\cdot \max\{d\geq 0\ s.t.\ \forall\,t'\in [t,t+d], \\
&&&\hspace{10em}\chi(p, \xi,t') = \chi(p, \xi,t)\}\\
&\theta^+(\neg\varphi,\xi,t)&=&\hspace{1em}-\theta^+(\varphi,\xi,t)\\
&\theta^+(\varphi_1\wedge\varphi_2,\xi,t)&=&\hspace{1em}\min(\theta^+(\varphi_1,\xi,t),\theta^+(\varphi_2,\xi,t))\\
&\theta^+(\varphi_1\mathbf{U}_I\varphi_2,\xi,t)&=&\hspace{1em}\max_{t'\in t+I} \min(\theta^+(\varphi_2,\xi,t'), \min_{t''\in [t,t')}\theta^+(\varphi_1,\xi,t''))
\end{align*}
As for space robustness, STL time robustness is sound in that positive or negative values of $\theta^+$ corresponds to satisfaction or violation in the usual Boolean interpretation.

\subsection{STL-based Resilience}\label{subsec:srs}
We now give an overview of the formulation of STL resilience introduced in~\cite{chen2022stl}. 
The intuition is that given an STL formula $\varphi$, two properties characterize its resilience: recoverability and durability. Recoverability requires a signal to recover from a violation of $\varphi$ within time $\Trec$; durability requires a signal to maintain the satisfaction of $\varphi$ for at least a duration of $\Tdur$. A signal is resilient if it satisfy both properties.
Given an STL formula $\varphi$ and $\Trec,\Tdur\in \mathbb{T}$, 
$\Tdur>0$, the formula 
\begin{align*}
    \srs_{\Trec,\Tdur}(\varphi) \equiv \neg \varphi\mathbf{U}_{[0,\Trec]}\mathbf{G}_{[0,\Tdur)}\varphi
\end{align*}
captures both temporal requirements (recoverability and durability). In~\cite{chen2022stl}, $\srs_{\Trec,\Tdur}(\varphi)$ expressions are the \emph{atomic formulas} of an STL-like logic called \textit{STL-based Resiliency Specifications}~(SRS). 

Akin to space and time robustness, a quantitative semantics in the form of a \emph{Resilience Satisfaction Value} (ReSV) is proposed to measure the resilience of SRS formulas. The ReSV of $\srs_{\Trec,\Tdur}(\varphi)$ w.r.t.\ $\xi$ at time $t$ is a $(\mathit{rec},\mathit{dur})$ pair given by:
\begin{align*}
    (-t_{rec}(\varphi,\xi,t)+\Trec, t_{dur}(\varphi,\xi,t)-\Tdur)
\end{align*}
where 
\begin{align*}
    t_{rec}(\varphi,\xi,t) &= \min \left( \{d \in \mathbb{T} \mid (\xi,t+d) \models \varphi \} \cup \{ |\xi|-t \}\right)\\[0.25ex]
    t_{dur}(\varphi,\xi,t) &= \min \left( \{d \in \mathbb{T} \mid (\xi,t'+d) \models \neg \varphi \} \cup \{ |\xi|-t' \}\right),\\
    &\hspace{2em}t'= t+t_{rec}(\varphi,\xi,t)
\end{align*}
The value of $t_{rec}(\varphi,\xi,t)$ quantifies the time needed for $\xi$ to recover from a violation of $\varphi$ at time $t$, and $t'=t+t_{rec}(\varphi,\xi,t)$ is the (absolute) recovery time. The value of $t_{dur}(\varphi,\xi,t)$ quantifies the time period after $t'$ during which $\varphi$ remains true. Thus, $\mathit{rec}$ tells us how early before $\Trec$ we recover, and $\mathit{dur}$ how long after $\Tdur$ $\varphi$ is maintained true. Going forward, we often abbreviate $t_{rec}(\varphi,\xi,0)$ and $t_{dur}(\varphi,\xi,0)$ with $t_{rec}$ and $t_{dur}$, respectively. 

Akin to space and time STL robustness, the authors of~\cite{chen2022stl} prove that the ReSV semantics is sound in that $(\xi,t)\models \srs_{\Trec,\Tdur}(\varphi)$ holds if $-t_{rec}(\varphi,\xi,t)+\Trec\geq 0$ and $t_{dur}(\varphi,\xi,t)-\Tdur\geq 0$, with at least one of the two inequalities strictly holding.
Thus, the resilience of $\varphi$ w.r.t.\ $\xi$ at time $t$ can be represented by a $(\mathit{rec},\mathit{dur})$ pair. 

The ReSV definition and the soundness result extend to composite SRS formulas. The intuition behind this extension is that the ReSV of e.g., an \textit{always} (\textit{eventually}) formula with bound $I$, represents the worst-case (best-case) resilience value attained by the subformula within $I$. For control purposes, however, we are only interested in SRS atoms.  See also Remark~\ref{rem:ctrl_obj}.


For finding an optimally resilient control strategy, it is necessary to compare the resilience of $\varphi$ w.r.t.\ two signals. In~\cite{chen2022stl}, an ordering relation $\succ_{re}$ is introduced specifically for this purpose.\footnote{The motivation for introducing the ordering relation $\succ_{re}$ in~\cite{chen2022stl} is a different one, namely for computing the semantics of composite SRS formulas.} The intuition is that usual Pareto-dominance $\succ$ over the reals is not consistent with resiliency satisfaction.  Recall that given two real-valued tuples $x,y\in \mathbb{R}^n$, $x$ \textit{Pareto-dominates} $y$, denoted by $x \succ y$, if $x_i\geq y_i$, $1\leq i\leq n$, and $x_i>y_i$ for at least one such $i$, under the usual ordering $>$. Now consider the $(\mathit{rec},\mathit{dur})$ pairs $(-1,2)$ and $(1,1)$.  By usual Pareto-dominance, $(-1,2)$ and $(1,1)$ are mutually non-dominated, but an ReSV of $(-1,2)$ indicates that the system doesn't satisfy recoverability; namely it recovers one time unit too late. On the other hand, an ReSV of $(1,1)$ implies satisfaction of both recoverability and durability bounds, and thus should be preferred to $(-1,2)$. This intuition is formalized in the definition of the $\succ_{re}$ relation. 


\begin{definition}[Resiliency Binary Relations~\cite{chen2022stl}]\label{def:binaryrelation}
We define binary relations $\succ_{re}$, $=_{re}$, and $\prec_{re}$ in $\mathbb{Z}^2$.  Let $x,y\in \mathbb{Z}^2$ with $x=(x_r,x_d)$, $y=(y_r,y_d)$, and \emph{sign} is the signum function. 
We have that $x \succ_{re} y$ if one of the following holds:
\begin{enumerate}
    \item $sign(x_r)+sign(x_d) = sign(y_r)+sign(y_d)$, and $x \succ y$.
    \item $sign(x_r)+sign(x_d) > sign(y_r)+sign(y_d)$.
\end{enumerate}
We denote by $\prec_{re}$ the dual of $\succ_{re}$. If neither $x\succ_{re} y$ nor $y \prec_{re} x$,\footnote{This is equivalent to saying that $sign(x_r)+sign(x_d) = sign(y_r)+sign(y_d)$ and neither $x\succ y$ nor $y \succ x$.} then $x$ and $y$ are \emph{mutually non-dominated}, denoted $x =_{re} y$.
Under this ordering, a \emph{non-dominated set} $S$ is such that $x=_{re} y$ for all $x,y \in S$. 
\end{definition}


Given a binary relation $\vartriangleright$ and a non-empty set $P\subseteq\mathbb{Z}^2$, we denote with $\max_{\vartriangleright}(P)$ the maximum set induced by the ordering $\vartriangleright$, i.e., the largest subset $S\subseteq P$ such that $\forall x\in S$, $\forall y\in P$, $x\not \vartriangleleft y$. The minimum set is defined analogously as $\min_{\vartriangleright}(P) = \max_{\vartriangleleft}(P)$. In the following, we will use the so-called \textit{maximum resilience set} $\max_{\succ_{re}}(P)$, abbreviated as $\max_{re}(P)$, and the one induced by canonical Pareto-dominance $\max_{\succ}(P)$, abbreviated as $\max(P)$.

\section{Problem Formulation}\label{sec:prob_form}

Consider a discrete-time, linear dynamical system $\system$ with dynamics
    $x_{t+1}=Fx_t+Gu_t$, 
where $F\in\mathbb{R}^{n\times n}$, $G\in\mathbb{R}^{n\times m}$, $x_t\in \mathbb{R}^n$ is the system state, and $u_t\in U\subseteq  \mathbb{R}^m$ is the control input at time $t$, with the control space $U$ defined as a closed polytope. 
Any sequence of control actions $\controls=[u_0,\ldots,u_{H-1}]$ induces a sequence of system states $\states = [x_0,\ldots,x_{H}]$ starting at $x_0$ and generated by the system dynamics. 

We now define the \emph{\probnamecap} problem as a \emph{bi-objective optimization} problem aimed at maximizing the recoverability and durability of $\states$ with respect to $\varphi$ at time $0$.\footnote{We slightly abuse notation and use sequences instead of signals.\vspace{-2.0ex}}

\begin{problem}[\probnamecap]\label{def:res_problem}
Let $\varphi$ be an STL formula, $\system$ the control system, $H$ the
control horizon, and $\Trec,\Tdur\in \mathbb{T}$, 
$\Tdur>0$. 
Solve
\begin{align*}
    \mathcal{S}^*
    = \underset{\controls}{max_{re}}(\Trec-t_{rec}(\varphi,\states,0),\;
    t_{dur}(\varphi,\states,0)-\Tdur)
\end{align*}
s.t.\ \ 
    $x_{t+1} = Fx_t+Gu_t, \ u_t\in U, \ t \in [0,\ldots,H-1]$.
\end{problem}

\begin{remark}
\label{rem:ctrl_obj}
Time $t=0$ represents the offset of\hspace{0.1em} $\states$ at which $\varphi$ is evaluated. The formulation of the problem is still general because we can consider trajectories starting at any state $x_0$.
The optimization objectives of Problem~\ref{def:res_problem} correspond to the ReSV semantics of $\srs_{\Trec,\Tdur}(\varphi)$ formulas, which is an atom in the SRS temporal logic of~\cite{chen2022stl}. In this way, we focus on optimizing the recoverability and durability of the \emph{first} recovery episode w.r.t.\ $\varphi$ over an MPC-style prediction horizon $H$~\cite{rawlings2000tutorial}. This is arguably more useful for control purposes than optimizing the ReSV of SRS formulas with temporal operators, where optimizing the ReSV of $\mathbf{G}_I \srs_{\Trec,\Tdur}(\varphi)$  (\,$\mathbf{F}_I \srs_{\Trec,\Tdur}(\varphi)$) corresponds to optimizing the worst-case (best-case) recovery episode w.r.t.\ $\varphi$ within $I$. From a technical perspective, the ReSV of a composite formula is itself a set of non-dominated $(\mathit{rec},\mathit{dur})$ pairs (as opposed to a single pair for $\srs_{\Trec,\Tdur}(\varphi)$ atoms), which would unnecessarily complicate the definition of Problem~\ref{def:res_problem}. 
\end{remark}

The \emph{optimal solution} to Problem~\ref{def:res_problem} is a set $\mathcal{S}^*\subseteq\mathbb{R}^2$ of non-dominated $(\mathit{rec},\mathit{dur})$ pairs, i.e., the maximum resilience set ($\max_{re}$) of all $(\mathit{rec},\mathit{dur})$ pairs induced by all possible sequences of control inputs. 
%
We denote by $\mathcal{U}_1^* \subseteq U^H$ the set of optimal points in the decision (control) space, where each point induces one optimal solution in $\mathcal{S}^*$. 

\begin{figure}
	\centering
	\includegraphics[width=\linewidth]{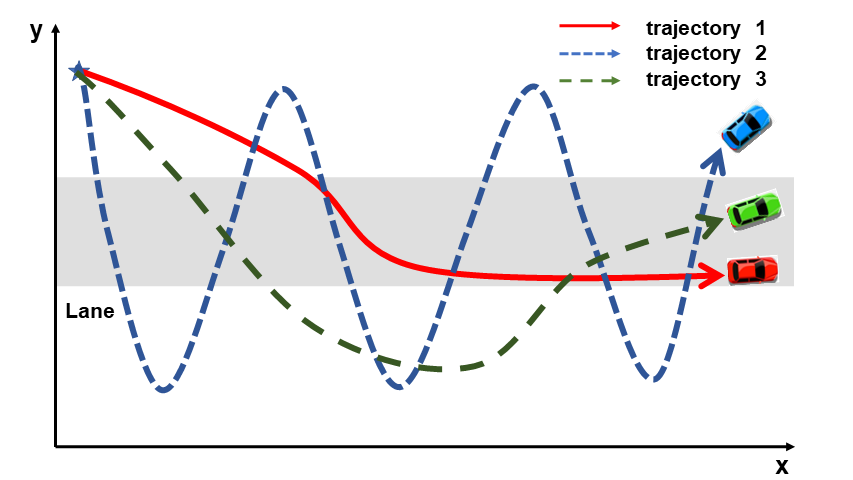}
	\caption{Vehicle trajectories corresponding to three optimal solutions provided by our
	\name framework.\hspace{0.5ex}   (The figure is for illustrative purposes only and doesn't directly reflect our experimental results.)}
	\label{fig:example1}
\end{figure}
\begin{example}
    In the lane-keeping problem, a vehicle is required to stay 
    within its lane (colored grey in Figure~\ref{fig:example1}) at all times.
    When the vehicle is forced to exit the lane due to e.g.\ an external disturbance, the \rescontroller\ predicts multiple optimal trajectories for the vehicle. Three such trajectories are shown in Figure~\ref{fig:example1}.  
    The initial location of the vehicle, marked by a star, is outside the lane,
    so it violates the lane-keeping specification $\varphi$. 
    Trajectory~1 is the slowest to recover from its  violation of $\varphi$, resulting in the worst recoverability among the three. However, the vehicle subsequently maintains $\varphi$ until the end of the trajectory, resulting in the best durability. In trajectory~3, the vehicle aggressively steers back into the lane whenever $\varphi$ is violated; so it exhibits the best recoverability. It cannot, however, maintain $\varphi$ after satisfaction due to overshooting; it thus has the worst durability. The behavior of trajectory~2 lies in-between trajectories~1 and~3. In our approach, the $(\mathit{rec},\mathit{dur})$ values of these three trajectories are mutually non-dominated and hence, equally resilient. 
\end{example}

Proposition~\ref{prop:maxset_parset} establishes the relationship between the sets of Pareto-optimal solutions obtained under the $\succ_{re}$ ordering and the usual $\succ$ ordering. This result is useful \emph{per se} and will be also required in Section~\ref{sec:solution_method} to prove the correctness of our algorithm.
\begin{proposition}\label{prop:maxset_parset}
    For any $P\subseteq\mathbb{Z}^2$ with $P\neq \emptyset$, $\max_{re}(P)\subseteq\max(P)$.
\end{proposition}
\begin{proof}
    We prove that any $x$ in $\max_{re}(P)$ is also in $\max(P)$. Let us remark that, if $x=(x_r,x_d)\in \max_{re}(P)$, then $x\succ_{re}y$ or $x=_{re}y$ for all $y=(y_r,y_d)\in P$. Similarly, if $x\in \max(P)$, then for all $y\in P$, $x\succ y$ or $x$ and $y$ are mutually non-dominated w.r.t.\ $\succ$.  We distinguish two cases:
        \textbf{(i)}~if $sign(x_r)+sign(x_d) = sign(y_r)+sign(y_d)$, we have $x\succ y$ (when $x\succ_{re} y$) or $x$ and $y$ are mutually non-dominated (when $x=_{re} y$).
        \textbf{(ii)}~if $sign(x_r)+sign(x_d) > sign(y_r)+sign(y_d)$, we show that $x\not \prec y$ (i.e., $x\succ y$ or $x$ and $y$ are mutually non-dominated). If, by contradiction, $y$ Pareto-dominates $x$ then we have $y_r\geq x_r$ and $y_d\geq x_d$ with at least one inequality holding strictly. This contradicts the assumption $sign(x_r)+sign(x_d) > sign(y_r)+sign(y_d)$ because $sign$ function is monotonic non-decreasing.
\end{proof}

The related problem of \emph{Time-Robust STL Control Synthesis}~\cite{rodionova2021time} seeks to maximize a single objective: the STL time robustness of $\states$ with respect to $\varphi$ at time $0$. 
For this reason, its optimal solution is a single value (if one exists).

\begin{problem}[Time-Robust STL Control Synthesis]\label{def:trobust_problem}
Let $\varphi$ be an STL formula, $\system$ the control system, and $H$ the control horizon.
Solve
\begin{align*}
    \theta^*
    = \underset{\controls}{max}\,\theta^+(\varphi,\states,0)
\end{align*}
subject to the constraints of Problem~\ref{def:res_problem} and $\theta^+(\varphi,\states,0)\geq \theta_l >0$.
\end{problem}
We note that the time robustness $\theta^+(\varphi,\states,0)$ is constrained by a positive lower bound $\theta_l$, meaning that the above problem is not always feasible. The result $\theta^*$ is the \emph{optimal solution} to the problem. We denote with $\mathcal{U}_2^*$ the corresponding set of \emph{optimal points} in the decision/control space.

\begin{remark}[Optimal points]
    In Problem~\ref{def:res_problem}, there are in general multiple optimal points in $\mathcal{U}_1^*$ for two reasons: (1)~$\mathcal{S}^*$ may contain multiple optimal solutions; (2)~even if $\mathcal{S}^*$ contains only one solution, there might be multiple control strategies inducing the same optimal solution.\footnote{Even if two different controllers generate two different trajectories, these trajectories might have the same recoverability and durability.}
    In Problem~\ref{def:trobust_problem}, there may be multiple points in $\mathcal{U}_2^*$ for a similar reason: it possibly contains multiple $\controls$ yielding the same optimal time robustness $\theta^*$. 
\end{remark}

\begin{proposition} 
Problem~\ref{def:res_problem} is solvable, i.e., an optimal point exists.
Problem~\ref{def:trobust_problem} is solvable if feasible.
\end{proposition}
\begin{proof}
Problem~\ref{def:res_problem} has a closed-polytope feasible region. 
Its objectives are integers and bounded, so there exists an optimal point such that the optimal values are achieved, hence solvable. Similar reasoning applies to Problem~\ref{def:trobust_problem}; thus it is solvable if feasible.
\end{proof}

Although the two problems seem  different, we emphasize that Problem~\ref{def:res_problem} generalizes Problem~\ref{def:trobust_problem} in the sense that solving the latter is equivalent to finding a particular optimal solution to the former.\footnote{With this result at hand, we will skip any experimental comparison between time-robust and resilient controllers, as the former is a special case of the latter.} See the following proposition.

\begin{proposition}\label{prop:include_solution}
    For a system $\system$, control horizon $H$, and STL formula $\varphi$,  we have $\mathcal{U}^{*}_2\subseteq\mathcal{U}_{1}^*$.
\end{proposition}
\begin{proof}

If Problem~\ref{def:trobust_problem} is infeasible, it is trivial that $\mathcal{U}^{*}_2=\emptyset\subseteq\mathcal{U}^{*}_1$. Otherwise, we prove that $\controls\in \mathcal{U}_{1}^*$ for all $\controls\in \mathcal{U}_{2}^*$. 
Because $\theta^*>0$, executing $\controls$ on $\system$ induces a trajectory where $\varphi$ is satisfied at time $0$ and $\varphi$ is maintained for as long as possible from time~$0$. 
This indicates that $\controls$ maximizes $\Trec-t_{rec}$ and $t_{dur}-\Tdur$ to their global maximum simultaneously, which is an optimal solution in $\mathcal{S}^*$.
Therefore, we have $\controls \in\mathcal{U}_1^*$.
\end{proof}

\section{Solution~method~for Resilient STL Control}\label{sec:solution_method}

In this section, we introduce our solution method for solving Problem~\ref{def:res_problem}. We note that both the STL Boolean semantics and the resiliency objectives $\Trec-t_{rec}$ and $t_{dur}-\Tdur$ are discrete (hence, non-smooth), which makes gradient-based methods unsuitable.
Meta-heuristics similarly tend to perform poorly and do not provide optimality guarantees.  For linear systems, however, prior work has shown that (single-objective) optimization of STL space and time robustness can be formulated and precisely solved as an MILP problem~\cite{raman2014model,rodionova2021time}.

Here, we take a similar approach and encode $t_{rec}$ and $t_{dur}$ using MILP constraints, building on the encoding of the Boolean STL semantics of Raman et al.~\cite{raman2014model}. To retrieve the full set of Pareto-optimal solutions, we define an $\epsilon$-constraint approach~\cite{laumanns2006efficient} that solves the bi-objective problem through multiple single-objective MILP instances, where one of the objectives is optimized and the other is constrained above some given level. 

Section~\ref{subsec:encode} presents our MILP encoding of the Boolean STL semantics and the resiliency objectives.  In Section~\ref{subsec:moo}, we present an $\epsilon$-constraint approach for efficiently computing the set of non-dominated optimal solutions, and provide a proof of its correctness. Section~\ref{subsec:complexity} analyzes our algorithm's computational complexity. 

\subsection{MILP Encoding}\label{subsec:encode}

The encoding method consists of the following three main steps.

\noindent
\textbf{(1)~Boolean semantics for STL atomic propositions.} Let $p=\mu(x_t)\geq c$ be an STL atomic proposition.
We use binary variables $z^\mu_t\in\{0,1\}$ to represent Boolean satisfaction ($z^\mu_t=1$) or violation ($z^\mu_t=0$) of $p$ over the control horizon at every time step $t=0,\ldots,H$. Assuming that the STL atomic propositions are linear w.r.t.\ $x_t$ (i.e., $\mu$ is a linear function), we can encode the Boolean semantics of STL atomic propositions with MILP constraints.
\begin{align}
       (z^\mu_t-1)\cdot M \leq\mu(x_t)-c\leq z^\mu_t\cdot M\label{equ:stl_predicate_encode}
\end{align}
where $M$ is a significantly large value.

\noindent
\textbf{(2)~Boolean semantics for STL composite formulas.}
The Boolean semantics for STL composite formulas are derived from STL atomic propositions using Boolean conjunction and disjunction. For a given STL formula $\varphi$, we introduce binary variables $z^\varphi_t$ to represent the Boolean semantics of $\varphi$ at time $t=0,\ldots,H$; i.e., $z^\varphi_t=1$ if $\varphi$ holds at time $t$, $0$ otherwise. 
The MILP encoding of $z^\varphi_t$ using only Boolean operators can be derived inductively~\cite{raman2014model}. 

\emph{Negation} $\varphi' = \neg\varphi$:  \begin{align}
    z^{\varphi'}_t=1-z^{\varphi}_t\label{equ:stl_negation_encode}
\end{align}

\emph{Conjunction} $\varphi=\bigwedge^m_{i=1}\varphi_i$:
\begin{equation}
\label{equ:stl_conjunction_encode}
\begin{aligned}
z^{\varphi}_t&\leq z^{\varphi_i}_t, i=1,\ldots,m\\
z^{\varphi}_t&\geq 1-m+{\textstyle\sum}_{i=1}^m z^{\varphi_i}_t
\end{aligned}
\end{equation}

\emph{Disjunction} $\varphi=\bigvee^m_{i=1}\varphi_i$:
\begin{equation}
\label{equ:stl_disjunction_encode}
\begin{aligned}
z^{\varphi}_t&\geq z^{\varphi_i}_t, i=1,\ldots,m\\
z^{\varphi}_t&\leq {\textstyle\sum}_{i=1}^m z^{\varphi_i}_t
\end{aligned}
\end{equation}

We now consider the encoding for STL formulas with temporal operators~\cite{raman2014model}. In particular, the \emph{always} and \emph{eventually} operators are respectively encoded as finite conjunctions and disjunctions using~\eqref{equ:stl_conjunction_encode} and~\eqref{equ:stl_disjunction_encode}. 
Below, we use the notation $a_t^H = \min(a+t,H)$ and $b_t^H = \min(b+t,H)$. Note that $a_t^H$ and $b_t^H$ are not additional MILP variables.


\emph{Always} $\varphi=\mathbf{G}_{[a,b]}\varphi_1$: we encode $z^{\varphi}_t$ as 
$\bigwedge^{b_t^H}_{i=a_t^H}z^{\varphi_1}_t$.


\emph{Eventually} $\varphi=\mathbf{F}_{[a,b]}\varphi_1$: we encode $z^{\varphi}_t$ as 
$\bigvee^{b_t^H}_{i=a_t^H}z^{\varphi_1}_t$.


\emph{Until} $\varphi=\varphi_1\mathbf{U}_{[a,b]}\varphi_2$: the satisfaction of $\varphi$ at $t$ can be derived from those of the following formulas, including an unbounded $\mathbf{U}$, to achieve a linear encoding w.r.t.\ $H$~\cite{bartocci2018specification}.
In particular, we encode $z^{\varphi}_t$ as $z^{\varphi'}_t$ given that $\varphi$ is equivalent to $\varphi'$, where 
\begin{equation}
\label{equ:stl_until_encode}
\varphi' = \mathbf{G}_{\left[0,a-1\right]}\varphi_1 \wedge \mathbf{F}_{\left[a,b\right]}\varphi_2 \wedge \mathbf{F}_{\left[a,a\right]}(\varphi_1\mathbf{U}\,\varphi_2).
\end{equation}
We note that $\varphi_1\mathbf{U}_{[a,b]}\,\varphi_2$ if $\varphi_1$ holds before $a$, after which $\varphi_1\mathbf{U}\,\varphi_2$ holds when $\varphi_2$ is satisfied before $b$. 
The first two conjuncts of~\eqref{equ:stl_until_encode} can be derived using the MILP encoding for the \emph{always} and \emph{eventually} operators. The unbounded \emph{until} in the last conjunct is encoded as follows~\cite{raman2014model}.
\begin{align*}
z^{\varphi_1\mathbf{U}\,\varphi_2}_t = z_t^{\varphi_2} \vee (z_t^{\varphi_1} \wedge z_{t+1}^{\varphi_1\mathbf{U}\,\varphi_2})
\end{align*}
for all $t=1,\ldots,H-1$, and $z^{\varphi_1\mathbf{U}\,\varphi_2}_H = z_H^{\varphi_2}$. 


\noindent
\textbf{(3)~Resilient STL control objectives.} 
Given an SRS expression $\srs_{\Trec,\Tdur}(\varphi)$, $\varphi$ an STL formula, we introduce variables $\recencoding$ and $\durencoding$ (and associated MILP constraints) to encode $t_{rec}(\varphi,\states,t)$ and $t_{dur}(\varphi,\states,t)$, respectively. 
Inspired by the encoding in~\cite{rodionova2021time}, $\recencoding$ and $\durencoding$ are defined as counters that, informally, keep track of the number of time units $\varphi$ remains violated and satisfied, respectively.   
\begin{align}
&\recencoding = (1-z^\varphi_t)\cdot(c^{\phi,rec}_{t+1}+1),
&c^{\phi,rec}_{H} = 0 \label{equ:c0t}
\end{align}
Variable $\recencoding$ is defined in reverse-temporal order; we first set $c^{\phi,rec}_{H}=0$. At time $t$, if $z^\varphi_t=1$ (i.e., $(\states,t) \models \varphi$ holds), we have $c^{\phi,rec}_t=0$; if $z^\varphi_t=0$ (i.e., $(\states,t) \models \varphi$ does not hold), we have $c^{\phi,rec}_t=c^{\phi,rec}_{t+1}+1$. 
Thus, if $\varphi$ doesn't hold at time $t$, $c^{\phi,rec}_t$ represents the time needed for $\varphi$ to recover (or the time until the end of $\states$ if $\varphi$ never recovers); or $0$ if $\varphi$ holds at $t$. We can see that $\recencoding$ follows exactly the the definition of $t_{rec}(\varphi,\states,t)$ in Section~\ref{subsec:srs}, whereby if $(\states,t) \models \varphi$ holds, we have $t_{rec}(\varphi,\states,t)=0$; otherwise, $t_{rec}(\varphi,\states,t)$ is the time needed for $\varphi$ to recover from violation.


To define $\durencoding$, we employ additional counter variables $c^1_t,c^2_t\in\mathrm{Z}_{\geq 0}$ for $t=0,\ldots,H$, which are similarly defined 
in reverse order as follows.
\begin{equation}
\label{equ:c1tc2t}
\begin{aligned}
&c^1_t = z^\varphi_t\cdot(c^1_{t+1}+1),
&c^1_{H} = 0\\
&c^2_t = (1-z^\varphi_t)\cdot(c^1_{t+1}+c^2_{t+1}),
&c^2_{H} = 0
\end{aligned}
\end{equation}
At time $t$, if $z^\varphi_t=1$, we have $c^1_t=c^1_{t+1}+1$, meaning that $c^1_t$ counts how many time units $\varphi$ remains true after $t$; if $z^\varphi_t=0$, i.e., $\varphi$ is false at $t$, we have $c^1_t=0$. 
At time $t$, if $z^\varphi_t=1$, we have $c^2_t=0$; if $z^\varphi_t=0$, we have $c^2_t=c^1_{t+1}+c^2_{t+1}$.
This variable keeps track, when $\varphi$ is false at $t$, for how long $\varphi$ will remain true after the next recovery episode. 

By the definition of $t_{dur}(\varphi,\states,t)$ in Section~\ref{subsec:srs}, if $(\states,t) \models \varphi$ holds, then $t_{dur}(\varphi,\states,t)$ is the time duration until a violation of $\varphi$ (or the end of $\states$); if instead $(\states,t) \not\models \varphi$,  $t_{dur}(\varphi,\states,t)$ refers to the duration-to-violation after the next recovery, and hence  remains constant until recovery.
We can see that $c^1_t,c^2_t$ respectively represent the behaviors of $t_{dur}(\varphi,\states,t)$ during satisfaction and violation of $\varphi$, thus
\begin{align}
    \durencoding = c^1_t+c^2_t\label{equ:tdur}
\end{align} 

\begin{remark}
    The MILP encoding for resilience objectives involves the multiplication of binary variables and (integer) counter variables, which is nonlinear. Nonetheless, we can convert them to MILP inequality constraints using the translation of the \emph{if-then-else} logic relation~\cite{bemporad2001discrete}. Let $z$ be a binary variable and $c$ an integer variable.  Then the relation $y = z\cdot c$ is equivalent to the following inequalities:
    \begin{equation}
    \label{equ:if_then_else}
        \begin{aligned}
            m\cdot z \leq & \;y\leq M\cdot z\\
        c-M\cdot(1-z)\leq & \;y\leq c-m\cdot(1-z)
        \end{aligned}
    \end{equation}
    where $M$ and $m$ are the upper and lower bounds for $c$, respectively.
\end{remark}

\begin{table}[t]
\centering
\resizebox{\columnwidth}{!}{%
\begin{tabular}{@{}lllllllllll@{}}
\toprule
$t$  & 0 & 1 & 2 & 3 & 4 & 5 & 6 & 7 & 8 & 9 \\ \midrule
$\bm{\chi(\varphi,\states,t)}$ & -1  & 1  & 1  & -1  & -1  & -1  & 1  & 1  & 1  & -1\\ \midrule
$z^\varphi_t$  & 0  & 1  & 1  & 0  & 0  & 0  &  1 & 1  & 1  & 0     \\ \midrule
$c^1_t = z^\varphi_t\cdot(c^1_{t+1}+1)$ & 0  & 2  & 1  &  0 &  0 & 0  & 3  & 2  & 1  &  0   \\ \midrule
$c^2_t = (1-z^\varphi_t)\cdot(c^1_{t+1}+c^2_{t+1})$ & 2  & 0  & 0  & 3  & 3  & 3  & 0  & 0  & 0  &   0   \\ \midrule
$\bm{\recencoding = (1-z^\varphi_t)\cdot(c^{\varphi,rec}_{t+1}+1)}$& 1  & 0  &  0 & 3  & 2  & 1  &  0 & 0  &  0 &  0 \\ \midrule
$\bm{\durencoding = c^1_t+c^2_t}$ & 2  & 2  & 1  & 3  &  3 &  3   &  3  &  2  &  1  &  0  \\ \bottomrule
\end{tabular}%
}
\caption{Example of MILP encoding  $c_t^{\varphi,rec}$,  $c_t^{\varphi,dur}$ of the resiliency objectives.}
\vspace{-4ex}
\label{tab:example1}
\end{table}

Table~\ref{tab:example1} provides an example of the encoding method. Consider $\states$ with $H=9$ and STL formula $\varphi$; characteristic function $\chi(\varphi,\states,t)$ is given in the table. A step-by-step computation of
$\recencoding$ and $\durencoding$ is provided; the results so obtained are the same as those for $t_{rec}(\varphi,\states,t)$ and $t_{dur}(\varphi,\states,t)$, computed using their definitions provided in Section~\ref{subsec:srs}. 
We define the function $$(\recencoding,\durencoding,\mathcal{C}_m)=\texttt{milp_encoding}(\varphi,\states)$$ that takes an STL formula $\varphi$ and a sequence of system states $\states$ as input, and outputs the set of encoded MILP constraints $\mathcal{C}_m$ and encoded variables $\recencoding,\durencoding$ using Eqs.~(\ref{equ:stl_predicate_encode})-(\ref{equ:tdur}). 


\subsection{Multi-Objective Optimization}\label{subsec:moo}

To address the challenge of multi-objective optimization, we propose an \emph{a posteriori}
method to reduce Problem~\ref{def:res_problem} to a sequence of single-objective MILP instances (by optimizing one of the objectives and constraining the other above some given level)   
and to efficiently generate the exact set of non-dominated optimal solutions. The single-objectives MILP instances can be solved by standard techniques such as branch-and-bound methods~\cite{lawler1966branch}. 

A particular property of the \probname\ problem is that both objective functions are discrete and bounded by the length of the horizon, so the number of optimal solutions is finite. Also, the optimal solutions are non-dominated with respect to the $\succ_{re}$ relation (which may not be the case with conventional Pareto dominance). We propose an algorithm for computing $\mathcal{S}^*$ taking these properties into account. 
First, we define the following problem.
\begin{flalign*}
    \hspace{-10em}P_{\epsilon}:\hspace{2em} \controls^*_\epsilon=arg\,\underset{\controls}{max}\; 
    t_{dur}(\varphi,\states,0)
    - \Tdur
\end{flalign*}
subject to the constraints of Problem~\ref{def:res_problem} and an additional \emph{$\epsilon$-constraint} $\Trec - t_{rec}(\varphi,\states,0) \geq \epsilon$. (One could equivalently maximize $\Trec - t_{dur}(\varphi,\states,0)$ and constrain $t_{dur}(\varphi,\states,0) - \Tdur\geq \epsilon$.) 
We respectively denote by $f_\epsilon^*$ and $g_{\epsilon}^*$ the values of $\Trec-t_{rec}(\varphi,\states,0)$ and $t_{dur}(\varphi,\states,0)-\Tdur$ corresponding to $\controls^*_\epsilon$. Our algorithm consists of the following steps.
\begin{enumerate}
    \item Let $\mathcal{S}^*=\emptyset$
    and $\epsilon=\Trec-(H-1)$. 
    \item If $P_\epsilon$ is feasible, go to step (3); otherwise, go to step (5).
    \item Solve $P_\epsilon$, then $\mathcal{S}^*=\mathcal{S}^*\cup\{(f_{\epsilon}^*,g_{\epsilon}^*)\}$. Set $\epsilon=f_{\epsilon}^*+1$.
    \item If 
    $\epsilon<\Trec+1$, go to step~(2); otherwise, go to step (5).
    \item Return $\mathcal{S}^*=\max_{re}\mathcal{S}^*$.
\end{enumerate}
The overall solution method is summarized in Algorithm~\ref{alg:overall}. 
To solve problem $P_{\epsilon}$, we encode it into MILP. To do so, we first generate the encoding for the control system through function \texttt{system}\_\texttt{constraints},  which takes as input the system $\system$ and decision variables $\controls$, and outputs the signal $\states$ (as a sequence of real variables) and the constraints $\mathcal{C}_s$ determined by the system dynamics. Second, we generate the encoding for the resiliency objectives through function \texttt{milp_encoding} as described in Section~\ref{subsec:encode}.

\begin{algorithm}[tb]
   \caption{Solution Method for \probnamecap}
   \label{alg:overall} 
\begin{algorithmic}[1]
    \INPUT{STL formula $\varphi$, control system $\system$, control\break horizon~$H$, time bounds $\Trec,\Tdur$.}
    \OUTPUT{The sets $\mathcal{S}^*$ and  $\mathcal{U}^*_1$ of Problem~\ref{def:res_problem}.}
    \STATE Initialize $\mathcal{S}^*=\emptyset$
    and $\epsilon=\Trec-(H-1)$. \WHILE{
    $\epsilon<\Trec+1$}
    \STATE Let $\controls$ be the decision variables.
    \STATE $(\states,\mathcal{C}_s)=\texttt{system_constraints}(\system,\controls)$.
    \STATE $(\recobj,\durobj,\mathcal{C}_m)=\texttt{milp_encoding}(\varphi,\states)$.
    \IF{$P_\epsilon$ is feasible}
    \STATE Solve $P_\epsilon$ as MILP and obtain $\controls_{\epsilon}^*$, $f_{\epsilon}^*$ and $g_{\epsilon}^*$.\label{line:solve_pe}
    \STATE $\mathcal{S}^*=\mathcal{S}^*\cup\{(f_{\epsilon}^*,g_{\epsilon}^*)\}$ and $\mathcal{U}^*_1=\mathcal{U}^*_1\cup\{\controls_{\epsilon}^*\}$.
    \STATE $\epsilon=f_{\epsilon}^*+1$.\label{line:eps_update}
    \ELSE
    \STATE $\epsilon=+\infty$.
    \ENDIF
    \ENDWHILE
    \STATE $\mathcal{S}^*=max_{re}\,\mathcal{S}^*$ and update $\mathcal{U}^*_1$ correspondingly.\label{line:max_res_set}
    \STATE \textbf{Return} $\mathcal{S}^*$ and $\mathcal{U}^*_1$.
\end{algorithmic}
\end{algorithm}

\begin{proposition}
    Algorithm~\ref{alg:overall} computes the exact set of optimal solutions $\mathcal{S}^*$ of Problem~\ref{def:res_problem} in a finite number of steps.
\end{proposition}

\begin{proof}
To prove the correctness of Algorithm~\ref{alg:overall}, we first prove that the points in $\mathcal{S}^*$ upon entry to step (5) above include all Pareto-optimal solutions according to the traditional ordering $\succ$. Then we prove that step (5) computes the exact set of optimally resilient solutions (according to $\succ_{re}$). 

To prove the first statement, it is enough to observe that at each iteration of the above while-loop, $f_{\epsilon}^*$ is strictly increasing and $g_{\epsilon}^*$ is non-increasing w.r.t.\ $\epsilon$, meaning that the $(f_{\epsilon}^*,g_{\epsilon}^*)$ pair at one iteration either dominates or is mutually non-dominated by the one at the previous iteration (according to $\succ$). Hence, $\mathcal{S}^*$ include all (but not necessarily only) the Pareto-optimal solutions according to $\succ$. 

For the second statement, we know that by Proposition~\ref{prop:maxset_parset}, the maximum resilience set of the Pareto front is equivalent to that of the whole solution space. Thus, by performing $\max_{re}\mathcal{S}^*$, the output of Algorithm~\ref{alg:overall} is the set of optimal solutions (according to $\succ_{re}$). 
Algorithm~\ref{alg:overall} terminates in a finite number of steps both because it requires solving at most $H$ instances of $P_{\epsilon}$, and each instance terminates in a finite number of steps. 
\end{proof}

\begin{proposition}\label{coro:num_pe}
        In the worst case, Algorithm~\ref{alg:overall} computes $H$ instances of problem $P_\epsilon$.
\end{proposition}

In the worst case, the number of increments of $\epsilon$ is $H$ (see line~\ref{line:eps_update} of Algorithm~\ref{alg:overall}), resulting in $H$ instances of solving $P_\epsilon$ at line~\ref{line:solve_pe}. 


\subsection{Computation Complexity}\label{subsec:complexity}

The computation complexity of Algorithm~\ref{alg:overall} consists of two major sources: the MILP problem and the multi-objective problem.
MILP problems are NP-hard and the computational complexity is highly dependent on the number of variables. 
In the worst case, a MILP problem solves a number of LP problems that are exponential in the number of binary and discrete variables. The complexity of LP is polynomial in the number of (real) variables. 

Let $\varphi$ be an STL formula with a set of atomic propositions $AP$. The Boolean semantics computation for STL atomic propositions introduces $O(H\cdot|AP|)$ binary variables; the Boolean semantics computation for $\varphi$ introduces $O(H\cdot|\varphi|)$ binary variables~\cite{raman2014model}. Hence, the number of binary and discrete variables is $O((|AP|+|\varphi|)\cdot H)$. The MILP encoding for the resiliency objectives introduces exactly $3\cdot H$ counter variables (i.e., $c^{\varphi,rec}_t,c^{1}_t,c^{2}_t$), hence this term is omitted. Note that the continuous variables are the sequences $\states\in \mathbb{R}^{n\cdot H}$ and $\controls\in\mathbb{R}^{m\cdot H}$. Since we need to solve at most $H$ MILP instances for $P_\epsilon$, then the overall complexity is $O(H\cdot(2^{(|AP|+|\varphi|)\cdot H} \cdot (H\cdot(m+n))^k))$ for some $k\geq 1$. 
Note that computing $\max_{re}\mathcal{S}^*$ in step (5) adds a cost quadratic in $H$ (see~\cite{chen2022stl}) and hence is negligible compared to the overall complexity.



\begin{remark}[Length of control horizon]
The length of the control horizon $H$ of the \probname problem should be carefully chosen. An excessively large $H$ introduces unnecessary computational complexity without significant performance improvement: since we optimize recoverability and durability relative to the first recovery episode, what happens after the durability period doesn't affect these two objectives. On the other hand, insufficiently large $H$ can make it difficult for the controller to provide an effective control action: if $\varphi$ is initially violated and $H$ is too small, it might be impossible to satisfy $\varphi$ within $H$, and so all control strategies will have the same objective values (worst possible $(\mathit{rec},\mathit{dur})$ values). 
\end{remark}

\section{Closed-loop control}\label{sec:mpc_dm}

In this section, we describe the remaining components of the \name framework: the MPC control strategy and DMs that selects a single solution from the set of optimal solutions.

\subsection{Model Predictive Control}

In the MPC setting, at each time step $t$, we solve Problem~\ref{def:res_problem} by setting $x_0$ to the current system state. 
The \rescontroller computes a set of optimal solutions $\mathcal{S}^*$. A 
DM selects a solution from $\mathcal{S}^*$ and then implements only the first step of the corresponding optimal control strategy. The control system is evolved following its dynamics. At time $t+1$, $x_0$ is set to the evolved system state; the next implemented control action is calculated similarly using the \rescontroller and a DM. This process is repeated at every remaining time step.

\subsection{Decision-Maker Design}

The optimal solutions $\mathcal{S}^*$ of Problem~\ref{def:res_problem} is a set of non-dominated $(\mathit{rec},\mathit{dur})$ pairs. At each step, an optimal solution is selected from $\mathcal{S}^*$ by
a DM.
We propose the following DM strategies representing different preferences in solution selection.
%
\textbf{\em Pro-recoverability DM}: selects the solution with maximum recoverability, representing a preference for rapid recovery.
%
\textbf{\em Pro-durability DM}: selects the solution with maximum durability, representing a preference for property maintenance post-recovery.
%
\textbf{\em Minimal-distance DM}: selects the optimal solution with minimal $L_2$-distance to the point $(\Trec, H-\Tdur)$ (the best attainable value of $(\mathit{rec},\mathit{dur})$).
%
\textbf{\em Adaptive DM}: respectively switches to pro-recoverability or pro-durability when maximum recoverability is less or greater than maximum durability. It represents a preference for the objective that is harder to achieve. 

We note that a DM strategy can also represent application-specific preferences beyond recoverability and durability (e.g., the average distance to the centerline of the lane in a lane-keeping problem). We leave this extension for future work.

\begin{remark}[a posteriori methods]
    The current design of \name uses an \emph{a posteriori} method: the complete set of optimal solutions $\mathcal{S}^*$ is first computed, then a DM choose one of them. However, we note that the above-defined minimal distance solution can be found without computing the Pareto front, but by solving the single-objective problem below:
    \begin{equation}\label{eq:min-distance-dm}
        s^*=\underset{\controls}{max}\;||(\Trec-t_{rec},t_{dur}-\Tdur),(\Trec,H-\Tdur)||
    \end{equation}
    where $\Trec, H-\Tdur$ are upper bounds on  $\Trec-t_{rec},t_{dur}-\Tdur$. 
    The solution $s^*$ is Pareto-optimal according to $\succ$ because, if it wasn't, there would exist $s'\succ s^*$ with $s'_1\geq s^*_1$ and $s'_2\geq s^*_2$, of which at least one is a strict inequality. Hence, $s'$ would be closer to $(\Trec,H-\Tdur)$ than $s^*$, which contradicts the fact that $s^*$ is the optimal solution of~\eqref{eq:min-distance-dm}.
    However, $s^*$ is not guaranteed to be an optimal solution to Problem~\ref{def:res_problem}, i.e.,  be Pareto-optimal according to $\succ_{re}$, unless we set $\Trec=0$ and $\Tdur=H$. We note that the latter is a perfectly reasonable choice for the bounds, representing the strictest possible requirements for both recoverability and durability.
\end{remark}

\section{Case Studies}\label{sec:casestudy}

In this section, we demonstrate the benefits of the STL-based resilient controller via two case studies. Experiments were performed on an Intel Core~i7-12700H CPU with 32GB of DDR5 RAM and a Windows~11 operating system. Our case studies have been implemented in MATLAB with YALMIP~\cite{Lofberg2004}; our implementation and case studies will be released in a publicly-available library.

\subsection{Lane Keeping}

We study resilient control in a lane-keeping problem.  We consider a linear, time-invariant single-track model for the vehicle with a constant nominal
longitudinal speed~\cite{mata2019robust}. The state-space representation of the model can be written as follows.
\begin{align*}
    \dot{\bm{x}}_{t}= \begin{bmatrix}
0 & 1 & 0 & 0\\
0 & a_{c1} & 0 & a_{c2}\\
0 & 0 & 0 & 1\\
0 & a_{c3} & 0 & a_{c4}
\end{bmatrix} \bm{x}_{t} + \begin{bmatrix}
0\\
\frac{2 C_{\alpha F}}{m}\\
0\\
\frac{2l_F C_{\alpha F}}{I_z}
\end{bmatrix} u_t,\hspace{0.6em} x_0 = [5, 6, 0, 2]^T
\end{align*}
where the state vector $\bm{x}_t=[y,{y}_v,\omega,{\omega}_v]^T$ with $y$ being the lateral position, ${y}_v$ the lateral velocity, $\omega$ the yaw angle, and ${\omega}_v$ the yaw velocity. Control actions $u_t$ are steering angles of the vehicle, which are bounded by the physical limitation of the vehicle: 
\begin{align*}
|u_t|\leq 0.72\;\text{rad},\hspace{2em} |u_t-u_{t+1}|\leq 0.72\;\text{rad}
\end{align*}
The parameters are defined as follows.
\begin{align*}
    a_{c1}&=-\frac{2C_{\alpha F}+2C_{\alpha R}}{m v},
    &a_{c2}&= -\frac{2l_{F}C_{\alpha F}-2l_{R}C_{\alpha R}}{mv}-v,\\
    a_{c3}&=-\frac{2l_{F}C_{\alpha F}-2l_{R}C_{\alpha R}}{I_{z}v},
    &a_{c4}&= -\frac{2l^2_{F}C_{\alpha F}+2l^2_{R}C_{\alpha R}}{I_{z}v}
\end{align*}
where $I_{z}$ is the inertial moment around the vehicle's $z$ axis; $l_{F}$ and $l_{R}$ are the distances between the CoG and the front and rear axles respectively. The constants $C_{\alpha F}$ and $C_{\alpha R}$ are front and rear cornering stiffness; $v$ is the constant nominal longitudinal speed. Our parameter selection is shown in Table~\ref{tab:car_params}. Letting $\Delta t=0.1$ secs be the length of one time-step, we have $\bm{x}_{t+1}=\bm{x}_{t}+\dot{\bm{x}}_{t}\cdot\Delta t$.

\begin{table}[t]
\centering
\begin{tabular}{@{}cccccccc@{}}
\toprule
        parameters & $l_F$ & $l_R$ & $C_{\alpha F}$ & $C_{\alpha R}$ &$I_z$ &$m$ &$v$\\ \midrule
values & 1.4 & 2.55 & 2200 & 2200 & 5757 & 2200  & 10 \\ \midrule
units & m & m & N/rad & N/rad &  kg/m$^2$ & kg & m/s \\\bottomrule
\end{tabular}
\caption{Selected vehicle parameters.}
\vspace{-6ex}
\label{tab:car_params}
\end{table}

\noindent\emph{Lane Keeping Property}: the vehicle should always remain within the lane boundaries in a time interval. 
\begin{align*}
    \varphi_{lk} \;=\; \mathbf{G}_{[0,h]}\,( p\wedge q) \;=\; \mathbf{G}_{[0,h]}((y_e\geq y_l) \wedge (y_e\leq y_u))
\end{align*}
where $y_e$ is the lateral difference between the vehicle and the center line of the lane. We set $y_u=1$ m, $y_l=-1$ m,
control horizon $H=60$, $h=2$, $\Trec=1.8$ secs, and $\Tdur=2.5$ secs. We apply \name to our vehicle model on a curvy track. 

We first evaluate the \name solution method by using Algorithm~\ref{alg:overall} to compute the optimal solutions at the initial step. The resulting optimal solutions are $\mathcal{S}^*=\{(-0.2,1.5),(0.1,-2.2),(0.2,-2.3)\}$; see Figure~\ref{fig:lk_traj}. In the top figure, the lane is indicated by the grey area and the starting location of the vehicle, marked by a star, is outside the lane. Each trajectory represents a predicted optimal trajectory for the vehicle and an optimal solution in $\mathcal{S}^*$. The middle figure shows the sequences of control actions for three optimal trajectories. The bottom figure shows the evolution of the lane-keeping requirement $\varphi_{lk}$ over time for the three optimal trajectories.

We now compare the different behaviors of the three optimal solutions. 
In the top figure, trajectory~1 represents a situation where the vehicle enters the lane the latest, and yet it remains in the lane till the end of the trajectory. We can also see that from the blue solid line in the bottom figure, $z^{\varphi_{lk}}$ recovers to~1 later than the others ($t_{rec}=2$ secs), but subsequently remains~1 for the longest duration ($t_{dur}=4$ secs); it thus results in the optimal solution $(-0.2,1.5)$.
Trajectory~3 represents a vastly different situation where the vehicle aggressively enters the lane first and stays in the lane, but quickly exits the lane due to overshooting. The yellow dashed line in the bottom figure reflects this situation: $z^{\varphi_{lk}}$
recovers to~1 the earliest ($t_{rec}=1.6$ secs) with the shortest subsequent duration ($t_{dur}=0.2$ secs), resulting in the optimal solution $(0.2,-2.3)$. 
Trajectory~2 represents an intermediate situation: the vehicle returns to satisfy $\varphi_{lk}$ with the second fastest recovery ($t_{rec}=1.7$ secs) and the second longest subsequent duration ($t_{dur}=0.3$ secs), yielding the optimal solution $(0.1,-2.2)$.

\begin{figure}[t]
	\centering
\includegraphics[width=\linewidth]{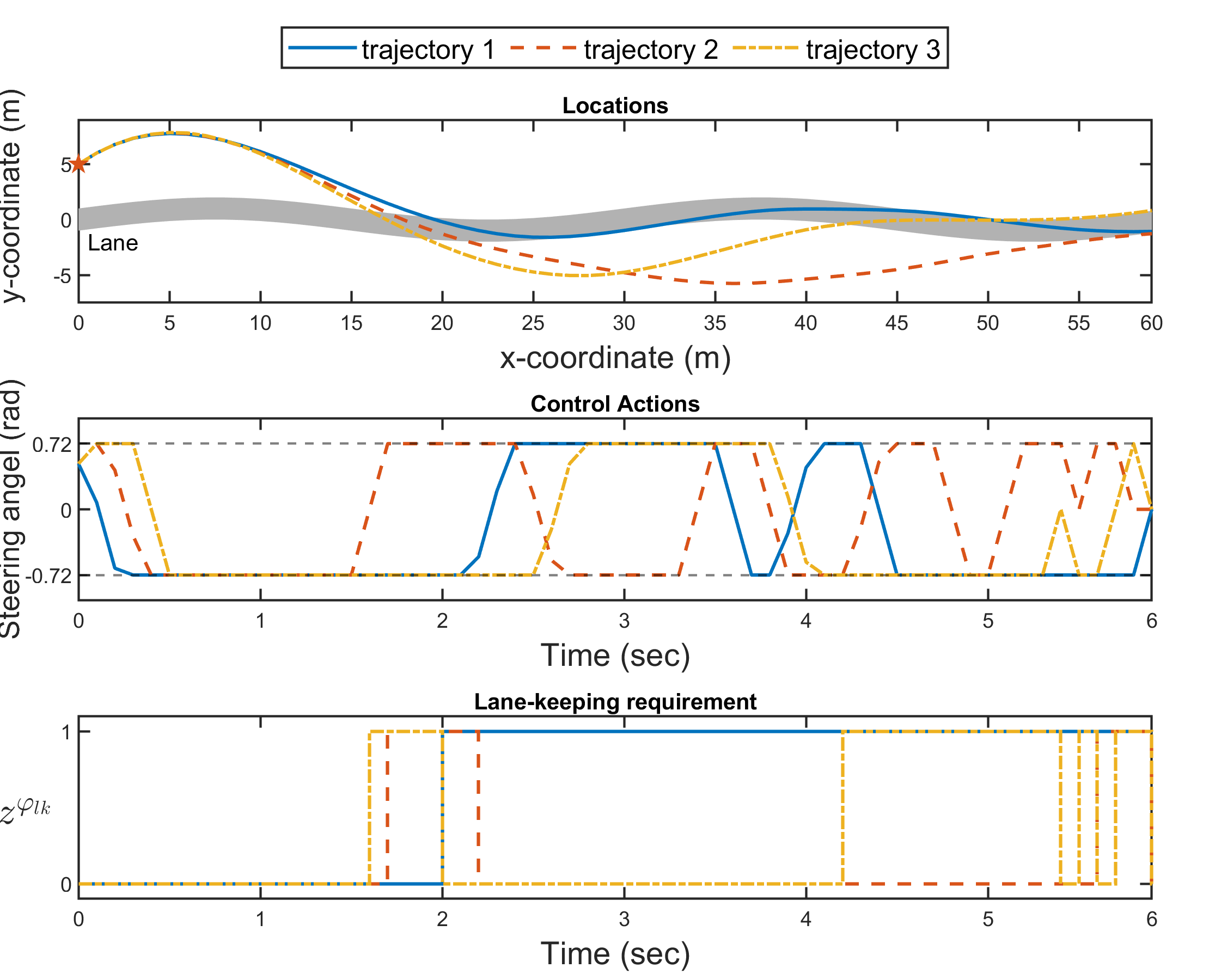}
    \vspace{-3ex}
	\caption{The optimal solutions provided by our \name framework at the initial step in a lane-keeping problem. }
    \label{fig:lk_traj}
    \vspace{-1ex}
\end{figure}

We then evaluate various DM strategies of our \name framework in the MPC setting. 
Solving Problem~\ref{def:res_problem} and selecting a control action took, on average, 78 msec on our machine.
We roll-out the MPC controller for a trajectory of 60 time-steps for each DM. The $(\mathit{rec},\mathit{dur})$ pair for the property $\varphi_{lk}$ for the pro-recoverability DM, pro-durability DM, adaptive DM, and minimal-distance DM are respectively: $(0.3,-2.1)$, $(-0.2,1.5)$, $(0,-1.9)$, and $(-0.2,1.5)$. See the results in Figure~\ref{fig:lk_mpc}.
As expected, the trajectory generated by the pro-recoverability DM has better recoverability yet worse durability compared to the pro-durability DM.
The trajectory generated with the adaptive DM has better recoverability compared to the pro-durability DM and better durability compared to the pro-recoverability DM, reflecting a balanced preference between recoverability and durability. The minimal-distance DM usually selects the same optimal solution as the pro-durability DM. This is because solutions with good recoverability often exhibit extremely bad durability due to overshooting, thus making them the farthest from the ideal resiliency value $(\alpha,H-\beta)$. 
This result evidences the usefulness of our approach in presenting the DM multiple, equally resilient, control strategies. In particular, we can see that optimizing for a fast recovery, which is roughly equivalent to maximizing STL time robustness, is not always the best strategy.


\begin{figure}[t]
    \centering
    \includegraphics[width=\linewidth]{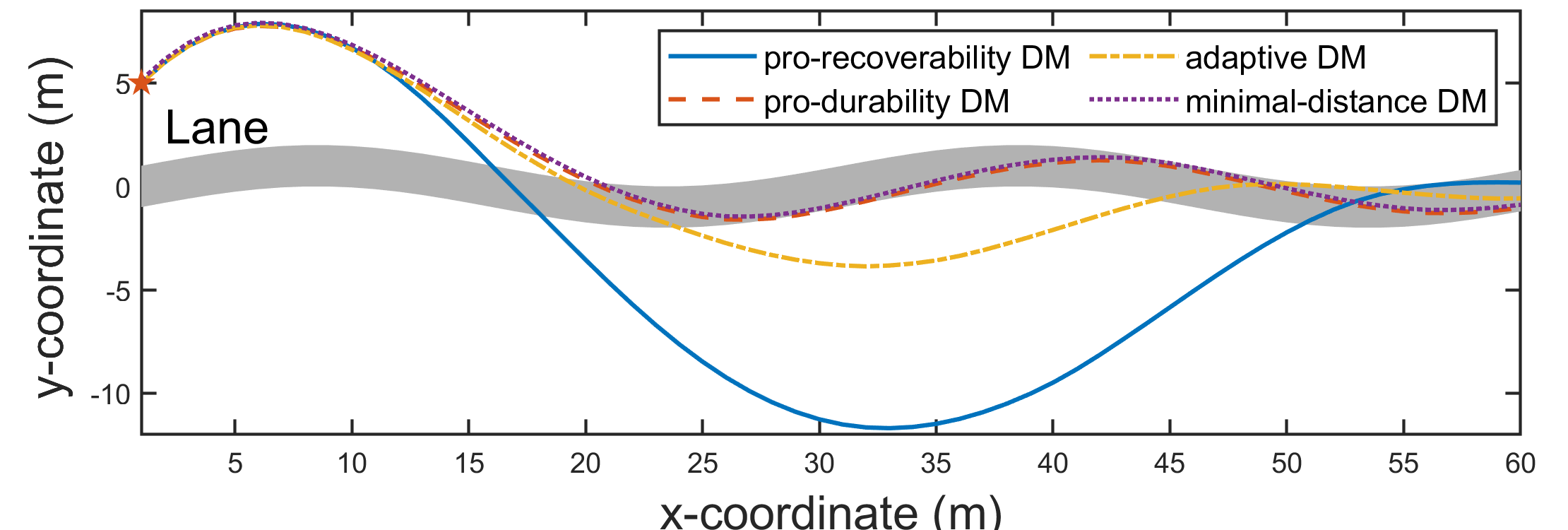}
    \vspace{-2ex}
    \caption{Simulated vehicle trajectories in \name.}
    \vspace{-2ex}
    \label{fig:lk_mpc}
\end{figure}


\subsection{Deadline-Driven Package Delivery}

We study a deadline-driven, multi-region cooperative package-delivery problem. The problem involves multiple 
controllable 
robots performing package deliveries by deadlines at multiple regions in a two-dimensional space. The robots are equipped with chargeable batteries. 

We extend a robot model in~\cite{rodionova2021time} with a battery state component. 
In an $N$-robot system, we denote the state vector of the $i$-th robot 
as $\bm{x}^i=[l_x^i, {v}^i_x,y^i,{l}^i_y,e^i,{v}^i_e]\in\mathbb{R}^6$, where $l^i_x,l^i_y$ are $x$, $y$ coordinates, ${v}^i_x,{v}^i_y$ are the $x$, $y$ velocities components, $e^i$ the battery level, and ${v}^i_e$ the battery charging rate. The full state vector of the multi-robot system is $\bm{x}=[\bm{x}^1,\ldots,\bm{x}^N]^T$.
Similarly, the control actions of the $i$-th robot are denoted by $\bm{u}^i=[u^i_1,u^i_2,u^i_3, e_{con}]$, where $u^i_1,u^i_2\in\mathbb{R}$ associate to coordinates, $u^i_3\in\{0,1\}$ indicates the charging status, and $e_{con}=-1$ is the battery \emph{consumption rate}; the control actions of the multi-robot system are $\bm{u}=[\bm{u}^1,\ldots,\bm{u}^N]^T$.
The state-space representation of an $N$-robot system at time $t$ is denoted as follows.
\begin{align*}
    \bm{x}_{t+1}= 
    F_{N}\cdot \bm{x}_{t} + G_N\cdot \bm{u}_{t} 
\end{align*}
where $\bm{x}_t$ and $\bm{u}_t$ are the system state and control actions, respectively. Matrices $F_N$ and $G_N$ are defined as follows.
\begin{align*}
    F_N &=I_N \otimes
    \begin{bmatrix}
       I_2\otimes A & 0 &0\\
       0 &  1 & t_s \\
       0 &    0 & 0 
        \end{bmatrix},\hspace{1em}
   A =  \begin{bmatrix}
1 & t_s \\
0 & 1 
\end{bmatrix}\\
G_N &=\begin{bmatrix} 
B_1 & 0 & \dots & 0\\
0 & B_2 &  \dots & 0 \\
\vdots & \vdots & \ddots & \vdots\\
0 & 0 & \dots & B_N 
\end{bmatrix},\hspace{.3em} B_i = \begin{bmatrix}
I_2\otimes b & 0 &0\\
0 &   d_u &0\\
0 &   e^i_{ch} &1\\
\end{bmatrix},\hspace{.3em} b = \begin{bmatrix}
 d_u \\   t_s
\end{bmatrix}
\end{align*} 
where $\otimes$ is the Kronecker product and $I_N$ is the identity matrix of size $N$. Parameter $e^i_{ch}$ is the battery \emph{charging rate} of the $i$-th robot, $t_s=0.1$ is the time-step size, and $d_u=0.005$.


\begin{figure*}[ht]
\centering
\subfloat[Trajectory 1]{\includegraphics[width=.5\textwidth]{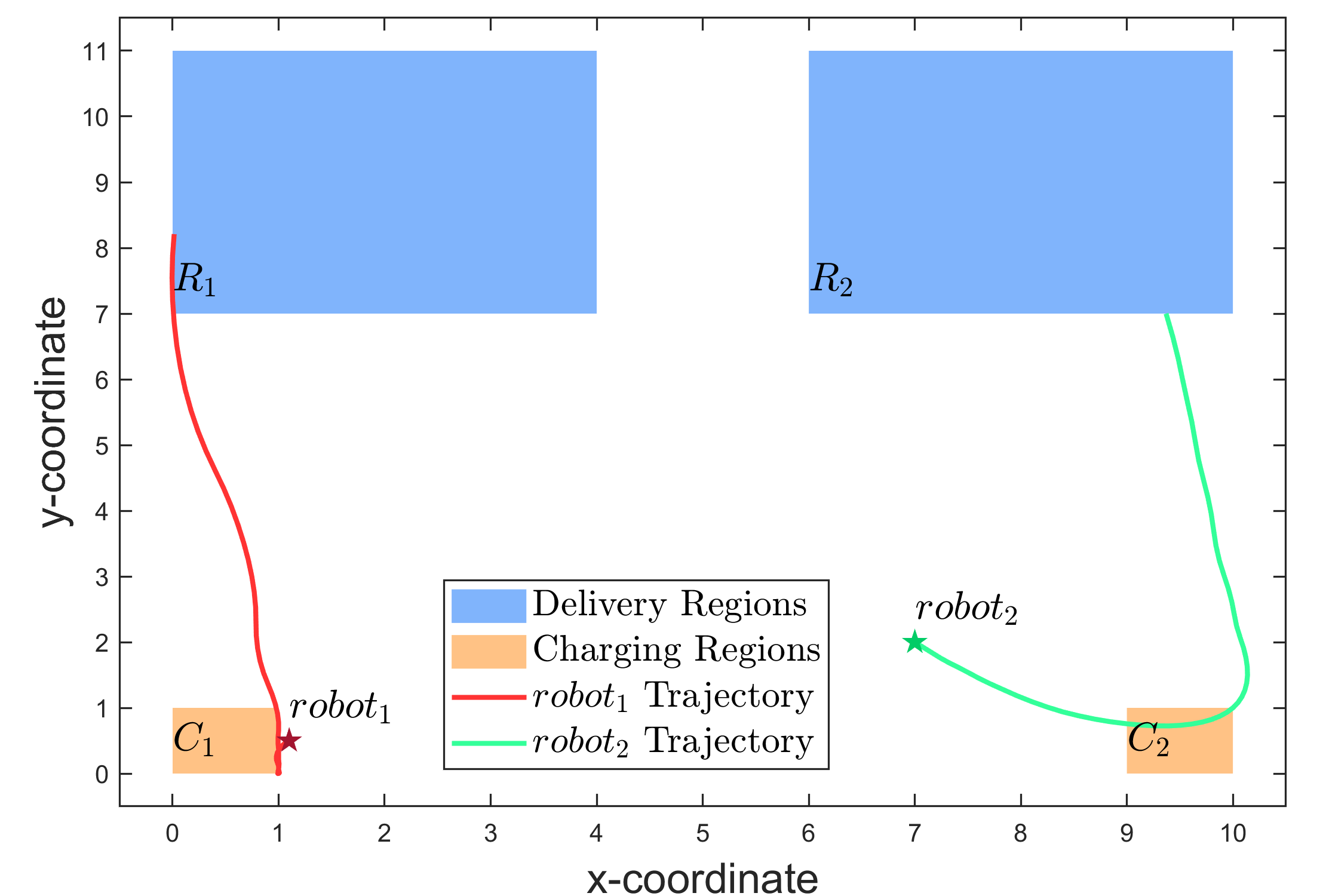}}
\subfloat[Trajectory 2]{\includegraphics[width=.5\textwidth]{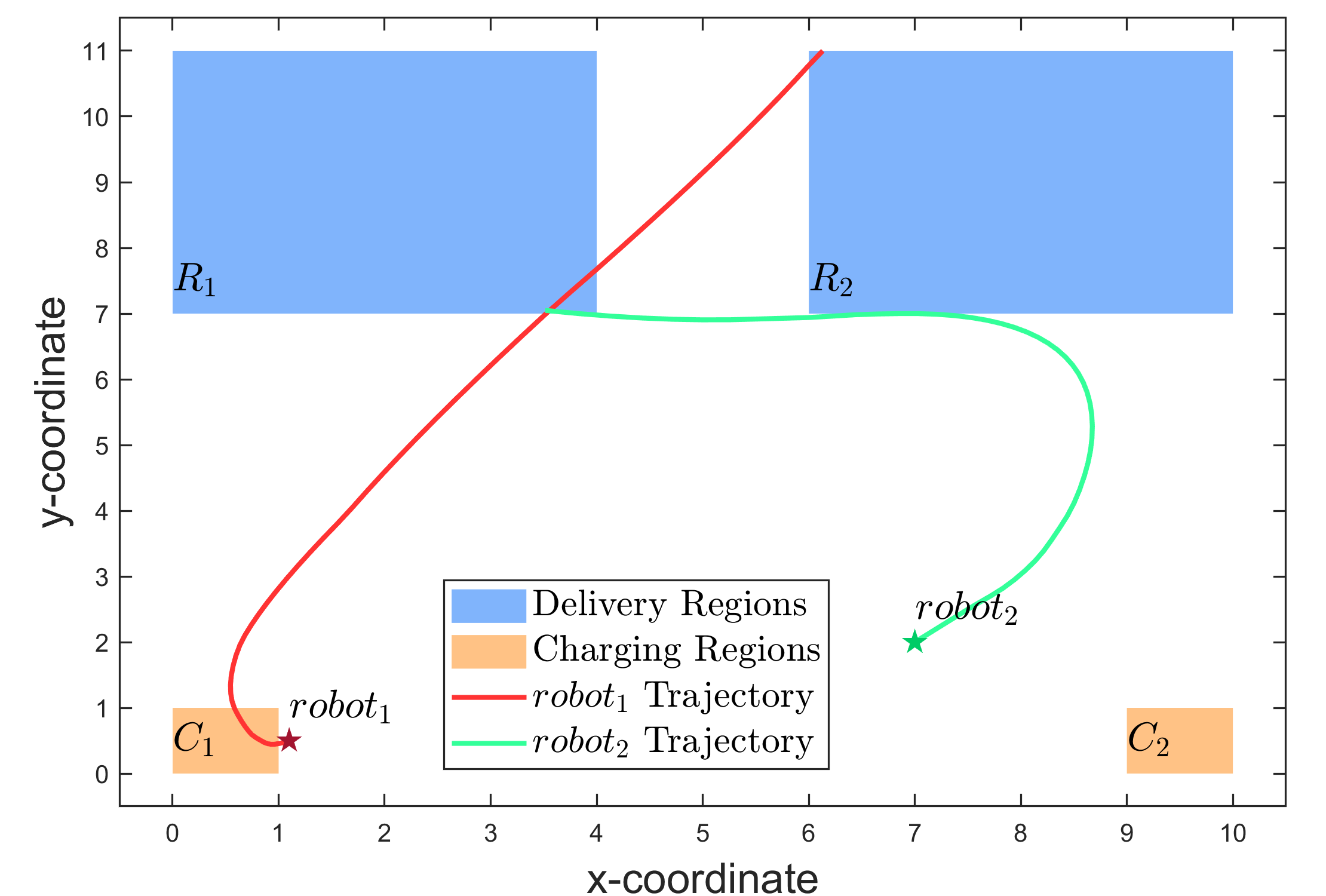}} 
\vspace{-1ex}
\caption{Optimal solutions provided by \name at the initial step in a deadline-driven, multi-region package-delivery problem.}
\vspace{-2ex}
\label{fig:mrs_traj}
\end{figure*}

\begin{figure*}[ht]
\centering
\subfloat[Trajectory 1: $\varphi_{del}$ has bad recoverability but good durability.]{\includegraphics[width=.49\textwidth]{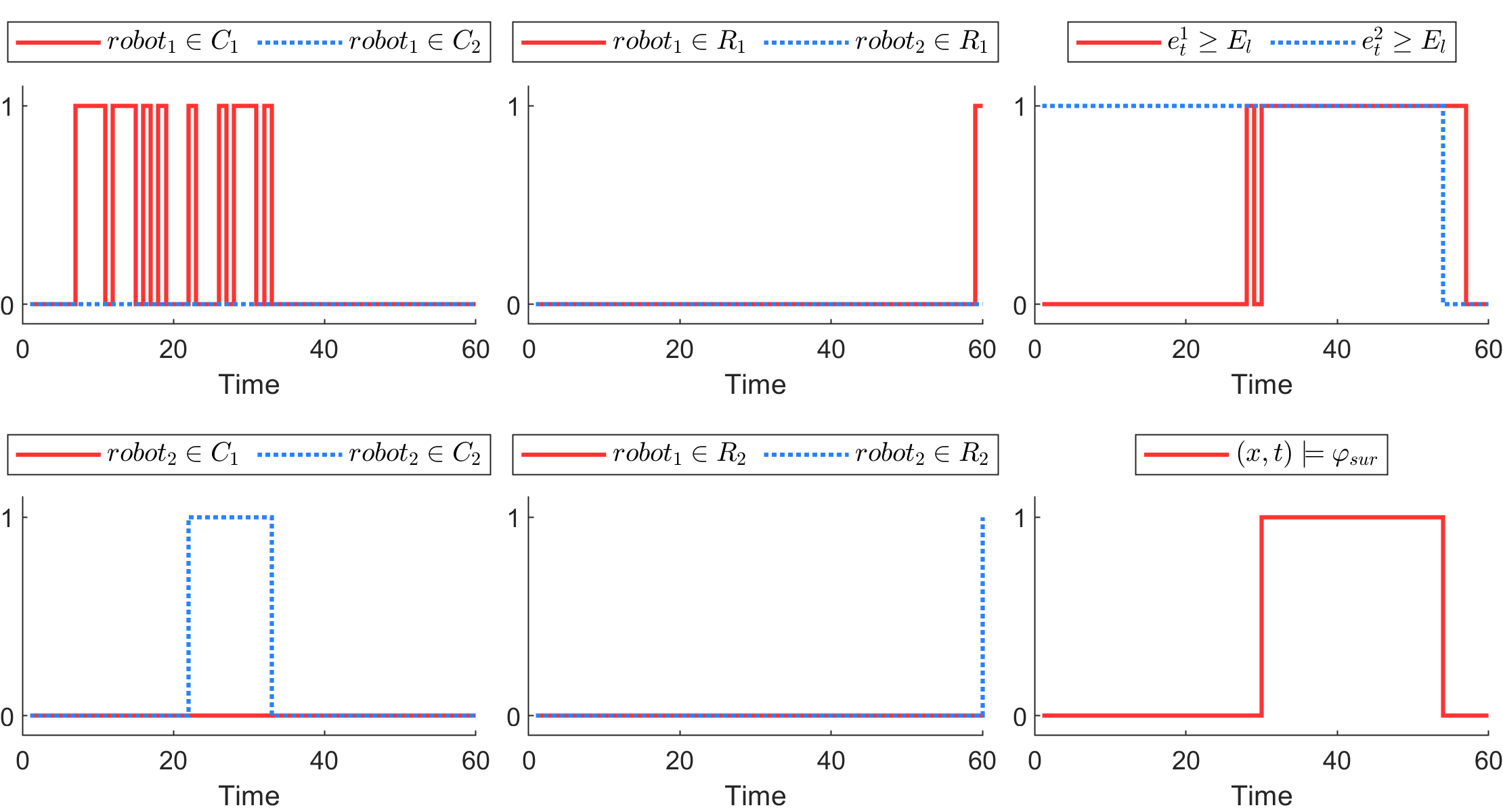}}\hfill
\subfloat[Trajectory 2: : $\varphi_{del}$ has good recoverability but bad durability.]{\includegraphics[width=.49\textwidth]{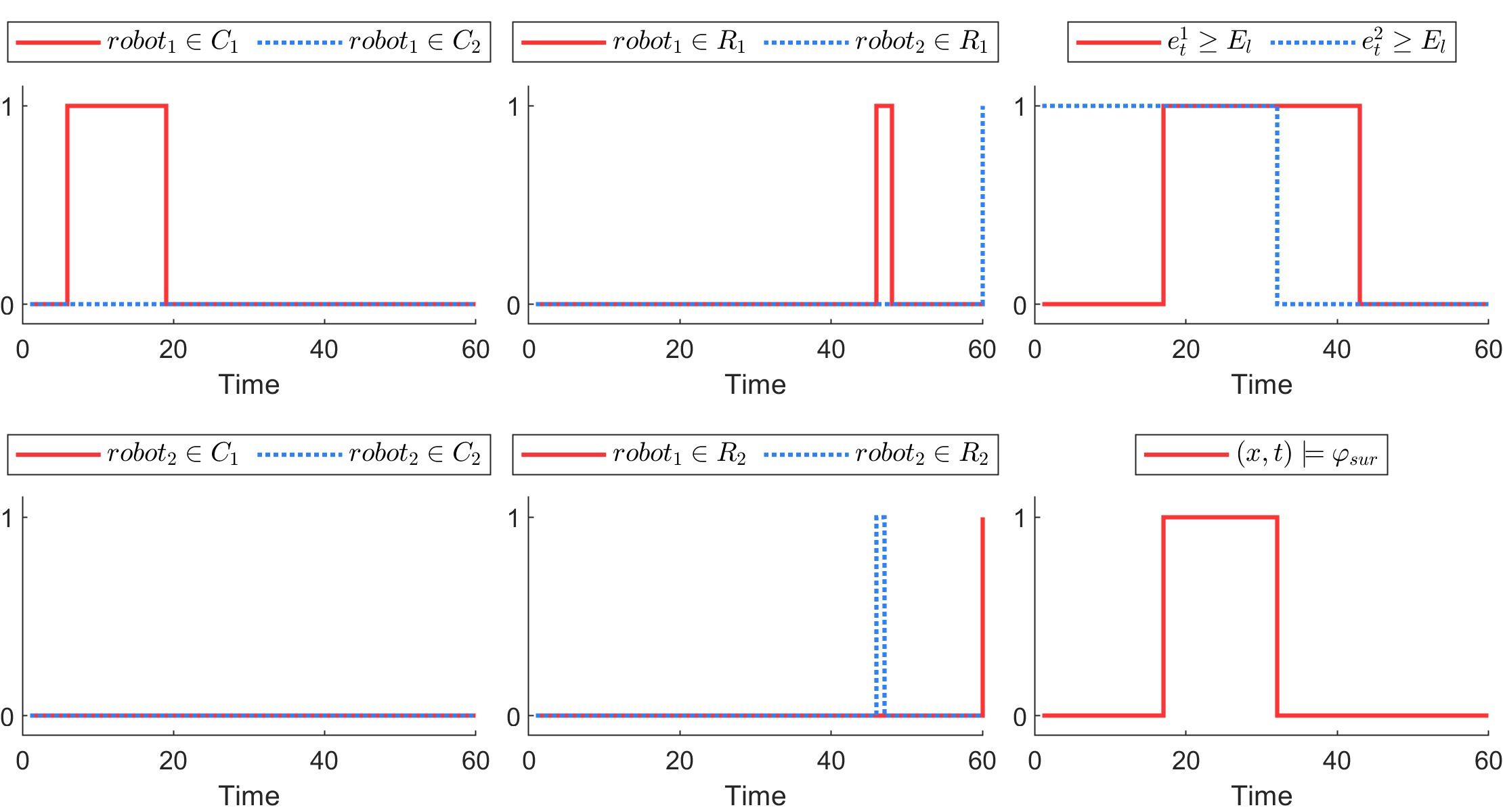}} 
\vspace{-1ex}
\caption{Evolution of the system states of two trajectories.}
\vspace{-2ex}
\label{fig:mrs_stl}
\end{figure*}

Consider two robots (i.e., $N=2$), denoted by $\mathit{robot}_1$ and $\mathit{robot}_2$, with $e^1_{ch}=7$ and $e^2_{ch}=2$ respectively representing an advanced and fast, and an outdated and slow, battery charging system.
\\
[0.2ex]
\emph{Deadline-driven Package Delivery}: A Package must be delivered by a robot by a deadline at a delivery region. 
Each of the two rectangular delivery regions, $R_1$ and $R_2$, has two deadlines defined by the time intervals of the $\mathbf{F}$ operators.
\begin{align*}
    \varphi_{s1} &=  \mathbf{F}_{[0,H/2]}\,(\mathit{robot}_1\in R_1 \vee \mathit{robot}_2\in R_1)\;\wedge\\
    &\hspace{8em}\mathbf{F}_{[H/2,H]}\,(\mathit{robot}_1\in R_1 \vee \mathit{robot}_2\in R_1)\\
    \varphi_{s2} &=  \mathbf{F}_{[0,H/2]}\,(\mathit{robot}_1\in R_2 \vee \mathit{robot}_2\in R_2)\;\wedge\\
    &\hspace{8em}\mathbf{F}_{[H/2,H]}\,(\mathit{robot}_1\in R_2 \vee \mathit{robot}_2\in R_2)
\end{align*}
where $\mathit{robot}_i\in R_j$ indicates the $i$-th robot is inside $R_j$, $i=1,2$, $j=1,2$. The requirements $\mathit{robot}_i\in R_j$ can be expressed by a set of four linear constraints.\footnote{For example, $\mathit{robot}_1\in R_1$ requires the $x^1$ to be greater than $R_1$'s lower bound on $x$, denoted by $x_{ub}$. Thus, we have $\bm{r}^1\cdot l\geq x_{ub}$, where $l=[1,0,0,0,0,0]^T$.} 
\\
[0.2ex]
\emph{Battery power requirement}: the battery power of the robots should remain above $E_{l}=10$.
\begin{align*}
    \varphi_{c1} = (e^1 \geq E_l),\hspace{2em}\varphi_{c2} = (e^2 \geq E_l)
\end{align*}
Robots' batteries can be charged in either of the two rectangular charging regions $C_1$ and $C_2$. 
Therefore, $u^1_3=1$ when $(\mathit{robot}_1\in C_1 )\vee (\mathit{robot}_1\in C_2)$ holds, and~0 otherwise; similarly, $u^2_3=1$ if and only if $(\mathit{robot}_2\in C_1) \vee (\mathit{robot}_2\in C_2)$ holds.
Constraints $\mathit{robot}_i\in C_j$ can be expressed as a set of four linear constraints similar to $\mathit{robot}_i\in R_j$. The overall requirement for the deadline-driven multi-region package delivery problem is defined as
\begin{align*}
    \varphi_{del} &= \varphi_{s1}\wedge\varphi_{s2}\wedge\varphi_{c1}\wedge\varphi_{c2}
\end{align*}

We set $H=60$, $\Trec=25$, $\Tdur=20$,  $x^1_0=[1.1,0,0.5,0,5,-1]$ and $x^2_0=[7,0,2,0,13,-1]$. The vectors specifying the lower and upper bounds on $x$ and those on $y$ of $R_1$, $R_2$, $C_1$, and $C_2$ are $[0,4,7,11]$, $[6,10,7,11]$, $[0,1,0,1]$ and $[9,10,0,1]$, respectively. We restrict the control actions $||u^i_1||,||u^i_2||\leq 1$ for $i=1,2$. 

We first compute the optimal solutions at the initial time-step in our \name framework using Algorithm~\ref{alg:overall}. The optimal solutions are $\mathcal{S}^*= \{(-4,4),(9,-5),(-3,1)\}$. Figure~\ref{fig:mrs_traj} shows the first two optimal solutions. Also, Figure~\ref{fig:mrs_stl} shows the evolution of the systems: in each sub-figure, from left to right, the top row shows the charging status of $\mathit{robot}_1$, package-delivery status at $R_1$, and the battery power level of robots; the bottom row shows the charging status of $\mathit{robot}_2$, package-delivery status at $R_2$, and the overall problem requirement $\varphi_{del}$.

In Figure~\ref{fig:mrs_traj}(a), an optimal situation needs $\mathit{robot}_2$ to go to $C_2$ for charging before package delivery at $R_2$ to ensure it has sufficient battery power for the deliveries. Meanwhile, $\mathit{robot}_1$ does not charge its battery in $C_1$ at full charging speed. This is because fast charging will not lead to quick satisfaction of $\varphi_{del}$ because of the slow package delivery at $R_2$ by $\mathit{robot}_2$. In Figure~\ref{fig:mrs_stl}(a), we can examine the system via the evolution of requirements: $\mathit{robot}_2\in C_2$ holds between $t=21$ and $t=33$, after which $\mathit{robot}_2\in R_2$ is true at $t=60$;
$\mathit{robot}_1\in C_1$ holds irregularly, after which $\mathit{robot}_1\in R_1$ is true at $t=59$. Overall, in the bottom-right figure, $\varphi_{del}$ is recovered late ($t_{rec}=29$), but remains true for a long period of time ($t_{dur}=24$); hence the solution $(-4,4)$.

In contrast, Figure~\ref{fig:mrs_traj}(b) depicts another optimal trajectory where $\mathit{robot}_2$ goes to $R_2$ and in turn $R_1$ for package deliveries without charging the battery, so to meet the deadlines. Meanwhile, $\mathit{robot}_1$ goes to $C_1$ and charges the battery at full charging speed to satisfy the battery power requirement as fast as possible; hence a quick recoverability w.r.t.\ $\varphi_{del}$. However, $\varphi_{del}$ does not remain true as long as in the first trajectory because $\mathit{robot}_2$ never charges the battery and thus its battery power quickly drops below $E_l$.  
Figure~\ref{fig:mrs_stl}(b) describes the system evolution: quick satisfaction of $\mathit{robot}_2\in R_2$, $\mathit{robot}_1\in R_1$, and $e^1\geq E_l$ collectively create the best recoverability of $\varphi_{del}$ ($t_{rec}=16$). However, even though two package deliveries at $R_1$ and $R_2$ meet the deadlines after recovery, $e^2\geq E_l$ cannot hold long enough because $\mathit{robot}_2\in C_1$ or $\mathit{robot}_2\in C_2$ never holds, causing the worst durability ($t_{dur}=15$). Hence the solution $(9,-5)$.

We then evaluate our DM strategies. We roll out the MPC controller for 60 steps and assess the recoverability and durability of the trajectories w.r.t.\ $\varphi_{del}$. 
Solving Problem~\ref{def:res_problem} and selecting a control action took, on average, 11.5 seconds on our machine.
The computational complexity is due in large part to the extensive nature of the STL requirements needed for this case study. We consider a strategy for reducing the execution time in Section~\ref{sec:conclusion}.
The $(\mathit{rec},\mathit{dur})$ pair of the trajectory with the pro-recoverability DM, the pro-durability DM, the adaptive DM, and the minimal-distance DM are respectively $(9,-5)$, $(-4,4)$, $(2,-1)$, and $(-4,4)$. As expected, the trajectories generated by the first three DM strategies respectively reflect a preference for recoverability, durability, and the recoverability-durability tradeoff. The minimal-distance DM shows a preference for good durability over good recoverability with overshooting.

\section{Related Work}\label{sec:related}

\textbf{Resilience in CPS}. 
Logic-based formulations of resilience in CPS have been proposed. The time robustness semantics for STL is considered equivalent to resilience in~\cite{mehdipour2021resilience}.  However, it can only quantify the recoverability of STL violations but not the subsequent durability.
Aksaray et al.~\cite{aksaray2021resilient} propose a ``time shifting'' STL and a resilient controller that maximizes the robustness value of the shifted formula as quickly as possible. This approach, however, does not consider the STL satisfaction durability post-recovery.
Resilient control frameworks include work by Bouvier et al.~\cite{bouvier2021quantitative} and Zhu et al.~\cite{zhu2011robust}.  Their non-logic-based notions of resilience, however, do not readily lend themselves to systems subject to diverse and sophisticated temporal requirements.
A survey on resilient multi-robot systems~\cite{prorok2021beyond} discusses how resilience is defined, measured, and maintained across various robotics domains.
Our work is based on the STL-based formulation of resiliency proposed by Chen et al.~\cite{chen2022stl}.  In this approach, the resilience of an STL formula takes into account both its recoverability and durability, which are quantified by sets of real-valued pairs.

\textbf{Control under STL specifications}.
In~\cite{belta2019formal}, a controller synthesis problem is solved to ensure that the behavior of the resulting control system satisfies the desired STL specifications. 
The STL robustness controller~\cite{raman2014model} uses a MILP encoding of an optimization problem to maximize the space robustness~\cite{donze2010robust} of the target STL specification. 
The controller synthesis problem for CPS subject to STL specifications has been considered in the context of
reactive control~\cite{raman2015reactive}, relaxed constrained MPC control~\cite{sadraddini2015robust}, and ``STL-based requirements priorities learning'' via robustness
slackness~\cite{cho2018learning}. 
Extensions to the original
robustness definition of STL specifications are used to tackle its disadvantages in optimization problems~\cite{lindemann2019robust,haghighi2019control}.

Instead of space robustness, the time-robust control problem~\cite{rodionova2021time} focuses on right-time robustness, which is critical in the presence of timing uncertainty. It maximizes the right-time robustness of an STL specification of a discrete linear system.
A left-right combined time robustness notion is proposed in~\cite{rodionova2022combined} to address the weakness of left- and right-time robustness for a single-directional time shift. It also proposes a control algorithm for linear systems that maximizes the combined time robustness using MILP.
An event-triggered MILP-based MPC framework~\cite{lin2020optimization} has been designed to maximize the overall space and time tolerances of the robustness degree of STL specifications for robot agents.

In contrast, we formally define the \probname problem for CPS with STL-based requirements as one that maximizes recoverability and durability, resulting in a multi-objective optimization problem.  
To the best of our knowledge, we are the first to consider a resilient control framework that co-maximizes recoverability and durability.

\section{Conclusion}\label{sec:conclusion}

We presented \emph{\name}, a resilient control framework for CPS subject to STL requirements. In \name, we defined the problem of resilient control as one of multi-objective optimization that maximizes both CPS  recoverability and durability w.r.t\ the desired STL properties. We proposed a solution method
that uses a MILP encoding and an a posteriori method for computing the precise set of non-dominated optimal solutions. Each optimal solution represents an optimally resilient trajectory of the control system. We also proposed a number of DM strategies that represent various preferences for selecting a single optimal solution. We illustrated \name on two case studies: lane keeping and deadline-driven multi-region package delivery. Collectively, our results showed the effectiveness of our solution methods in achieving resilient control, and demonstrated the effects of DM preferences. 


Future work will consider application-specific
DM strategies that go beyond recoverability and durability; e.g., in the lane-keeping problem, the average
distance to the centerline of the lane.
We will also investigate learning a \emph{neural controller} (NC) for our \name setup. The controller presented in this paper can be run repeatedly in simulation mode to provide the training data for the NC. Approaches of this nature can be found in~\cite{stoller2020neural,chen2021mpc}. Such an NC is expected to improve upon the execution time of the MPC controller (which has to solve a multi-objective MILP problem at every time step) for the package-delivery case study by orders of magnitude. 
\balance
\bibliography{mybibfile}

\end{document}